\documentclass[11pt]{article}
\usepackage{tikz}
\usepackage{amsmath, amssymb, amsthm}
\usepackage{graphicx,psfrag}
\usepackage{mathtools}
\usepackage{bm}
\usepackage{bbm}
\usepackage{enumerate}
\usepackage{natbib}
\usepackage{url} 
\usepackage[bottom]{footmisc}
\usepackage{booktabs}
\usepackage{multirow}
\usepackage{color}
\usepackage[hypertexnames=false]{hyperref}
\hypersetup{
colorlinks = true,
allcolors = blue
}
\usepackage{caption}
\usepackage[a4paper, 
            left=1.1in,right=1.1in,top=1in,bottom=1.2in,%
            footskip=.7in]{geometry}
   
\usepackage[ruled,vlined]{algorithm2e}   
\usepackage[toc,page,header]{appendix}
\usepackage{enumitem}
\usepackage{adjustbox}
\def\T{{ \mathrm{\scriptscriptstyle T} }}

\def\din{d_{\mathrm{in}}}              
\newcommand{\lmin}{\lambda_{\mathrm{min}}}
\newcommand{\lmax}{\lambda_{\mathrm{max}}}

\newcommand{\hB}{\hat{B}}

\newcommand{\argmin}{\textrm{argmin}}

\newcommand{\RtoR}{\mathrm{R2R}}

\def\post{\pi_n}

\def\Pa{\mathrm{PA}}
\def\Pred{\mathrm{PRED}}

\def\Des{\mathrm{DES}}
\def\NonDes{\text{NON-DES}}

\def\cG{\mathcal{G}}
\def\cN{\mathcal{N}}

\def\bbD{\mathbb{D}}
\def\bbS{\mathbb{S}}

\def\bbR{\mathbb{R}}
\def\ind{\mathbbm{1}}   
\def\T{\mathrm{T}}   

\def\sX{\mathsf{X}}

\def\lupp{\bar{\ell}}
\def\llow{\underline{\ell}}

\def\diag{\mathrm{diag}}
\def\trace{\mathrm{tr}}
\def\det{\mathrm{det}}
\newcommand{\Var}{\mathrm{Var}}
\SetKwInput{KwInput}{Input}                
\SetKwInput{KwOutput}{Output} 

\newtheorem{lemma}{Lemma}
\newtheorem{corollary}{Corollary}

\newtheorem{theorem}{Theorem}
\newtheorem{definition}{Definition}
\theoremstyle{definition}

\theoremstyle{remark}
\newtheorem{remark}{Remark}

\title{Identifiability of the minimum-trace directed acyclic graph and hill climbing algorithms without strict local optima under weakly increasing error variances}
\usepackage{authblk}
\author{Hyunwoong Chang\thanks{These authors contributed equally. Corresponding to hwchang@utdallas.edu} \, and Jaehoan Kim$^\dagger$}
\date{}
\affil{$^{*}$Mathematical Sciences, University of Texas at Dallas\\$^{\dagger}$Statistical Science, Duke University}

\begin{document}
\maketitle

\begin{abstract}
We prove that the true underlying directed acyclic graph (DAG) in Gaussian linear structural equation models is identifiable as the minimum-trace DAG when the error variances are weakly increasing with respect to the true causal ordering. This result bridges two existing frameworks as it extends the identifiable cases within the minimum-trace DAG method and provides a principled interpretation of the algorithmic ordering search approach, revealing that its objective is actually to minimize the total residual sum of squares. On the computational side, we prove that the hill climbing algorithm with a random-to-random (R2R) neighborhood does not admit any strict local optima. Under standard settings, we confirm the result through extensive simulations, observing only a few weak local optima. Interestingly, algorithms using other neighborhoods of equal size exhibit suboptimal behavior,
having strict local optima and a substantial number of weak local optima.
\end{abstract}
\noindent
\textit{Keywords:}
\small
Bayesian network; Causal discovery; Directed acyclic graph; Hill climbing algorithm; Identifiability; Structural equation model.

\section{Introduction}
\normalsize

We consider the problem of structure learning for a directed acyclic graph (DAG). A fundamental challenge is that the true underlying DAG cannot be uniquely identified from observational data alone, since multiple DAGs can encode the same set of conditional independencies~\citep{koller2009probabilistic}. This identifiability issue  causes computational inefficiency~\citep{chickering2002optimal} and complicates interpretation after estimation. To address this problem, a substantial body of research introduces additional assumptions on the underlying distributions that make the true DAG identifiable~\citep{shimizu2006linear, hoyer2008nonlinear, peters2011identifiability}.

In particular, \cite{peters2014identifiability} prove that the true DAG is identifiable in Gaussian linear structural models when all error variances are equal, which has inspired two lines of research. One line of work leverages the observation that 
source variables always have the smallest marginal variance~\citep{chen2019causal}, or that leaf variables always have the largest conditional variance given other variables~\citep{ghoshal2018learning}. These methods identify the true node ordering by finding and removing a source or leaf variable, then recursively applying the procedure to the induced subgraph. \citet{park2020identifiability} show that the true DAG can be identified by these methods in extended settings where error variances weakly increase in the true causal ordering of variables. 
Although these methods are efficient, they lack probabilistic interpretability, making it challenging to compare the likelihood of different orderings and limiting their extension to probabilistic frameworks such as Bayesian models. 
On the other hand, leveraging the observation that the sum of error variances is minimized at the true DAG under the equal error variance, \cite{aragam2019globally} propose the minimum-trace DAG, which minimizes the total residual sum of squares across all variables. Interestingly, a Bayesian formulation with a prior enforcing equal error variances yields a posterior distribution that corresponds to an objective function that targets minimum-trace DAGs~\citep{chang2024order}.
However, the identifiability issue remains, as multiple minimum-trace DAGs may exist in general settings. Prior to this work,  the equal error variance condition was the only known identifiable case within the minimum-trace DAG framework.

This paper reconciles these two aspects by proving that the true DAG is identifiable as the minimum-trace DAG when error variances are weakly increasing with respect to the true causal ordering. This result is nontrivial, especially given that the equal error variance case corresponds to a measure-zero subset under the Lebesgue measure over the parameter space~\citep{spirtes2000causation}, whereas the weakly increasing variance condition holds on a set of positive measure.
Additionally, the assumption of weakly increasing error variances may be reasonable in practice, as upstream variables often play more foundational roles, while descendant variables tend to reflect accumulated uncertainty.

From a computational perspective, the minimum-trace DAG method naturally adopts order-based algorithms, which search over the ordering space to find an optimal ordering that maximizes a given on objective functions~\citep{teyssier2005ordering, scanagatta2017improved}. 
While significantly smaller than the space of  DAGs, the space of orderings (i.e., the permutation space) still has $p!$ elements for $p$ variables, posing substantial computational challenges. The simplest among these algorithms is the hill climbing algorithm, a greedy local search method that iteratively moves to a better solution within a predefined local neighborhood, terminating when no further improvement is possible. This is also closely related to the
Metropolis-Hastings algorithms for Bayesian models as they require predefined local neighborhoods to construct proposal distributions~\citep{friedman2003being, agrawal2018minimal, chang2024order}. In general, the choice of the local neighborhood significantly influences the efficiency of the algorithms. 
If the neighborhood is too small, it is more likely to get trapped in local optima, whereas an excessively large neighborhood incurs substantial computational cost at each iteration. 
Despite its importance, the choice of local neighborhood structure is largely heuristic and lacks theoretical justification. 

We prove that the hill climbing algorithm with a random-to-random (R2R) neighborhood does not admit any strict local optima. 
Under simulation settings commonly used in the literature, no strict local optima and only four occurrences of weak local optima were observed across 10,000 simulation replications, suggesting that 
the result holds in practice.
More intriguingly, we compare the R2R neighborhood with two other neighborhoods of the same size, namely, random transposition (RTS) and reversed random-to-random (R2R-REV), which exhibit some strict local optima and numerous weak local optima. This result suggests that the R2R neighborhood is arguably an optimal choice, indicating that the optimality of a neighborhood can depend not just on its size but also on its scheme. This perspective alone might be of independent interest, potentially alluding to open problems in combinatorics. 

\section{Model identifiability under weakly increasing error variance}

\subsection{Preliminary}

\indent A DAG $G$ is a pair $(V, E)$ where $V$ is the vertex set and $E \subset V \times V$ is the set of directed edges. Throughout the paper, we assume $V = [p] = \{1, \dots, p \}$ for DAG models, used to index random variables $\sX_1, \ldots, \sX_p$. We assume that $E$ adheres to the structure of a DAG, meaning it contains no undirected edges or cycles. We use the notation $i \rightarrow j \in G$ to mean that $(i, j) \in E$. Let $|G|$ denote the number of edges in the DAG $G$, that is, $|G| = |E|$. Given a directed acyclic graph $G$, we denote the parents of a node $j$ by $\mathrm{PA}_j(G)$. 
An ordering $\sigma$ is
a permutation $\sigma : [p] \to [p]$, where $\sigma$ yields a causal ordering for $G$ if and only if, for any indices $i < j$, $\sigma(i)$ is not a descendant of $\sigma(j)$ in $G$. 
We denote \( \mathbb{S}^p \) as the set of all orderings of \( p \) indices. 
Let $\Pred_j(\sigma) = \{\sigma(k): k < \sigma^{-1}(j) \}$ denote the set of predecessors of node $j$ under the ordering $\sigma$,  that is, all nodes that appear before $j$ in $\sigma$. 
Let $\mathcal{G}(\sigma)$ denote the collection of all DAGs are consistent with an ordering $\sigma$; that is, $\mathcal{G}(\sigma) = \{G: \sigma \text{ is a causal ordering for  } G \}$. We denote the univariate normal distribution and \(d\)-dimensional multivariate normal distribution by \(\mathrm{N}\)  and \(\mathrm{MVN}_d\), respectively.


We consider a structural equation model with the true causal DAG $G^*$, 
where the data are generated as follows.
\begin{align}\label{eq:str_eq_individual}
    \sX_j = \sum_{i \in \mathrm{PA}_j(G^*)} B^*_{ij} \sX_i + \epsilon_j,
\end{align}
where an error term $\epsilon_j \sim \mathrm{N}(0, \omega^*_j)$. We refer to the coefficient matrix $B^* = (B^*_{ij}) \in \bbR^{p \times p} $ as the weighted adjacency matrix, as each nonzero entry $B^*_{ij}$  indicates the presence of a directed edge $i \rightarrow j$ in $G^*$.
Let $[\sigma^*] = \{\sigma \in \bbS^*:  G^* \in \cG(\sigma)\}$ denote the set of all orderings consistent with $G^*$, where $\sigma^*$ is an element in $[\sigma^*]$.
Let $\sX =(\sX_1, \dots, \sX_p)^\T$ be a vector of $p$ random variables, and express~\eqref{eq:str_eq_individual} in matrix form as $\sX = (B^*)^\T \sX + \epsilon$, where $\epsilon \sim \mathrm{MVN}_p(0, \Omega^*)$ with $\Omega^* = \diag (\omega^*_1, \dots, \omega^*_p)$. Under the form, one can readily verify that $\sX \sim \mathrm{MVN}_p(0, \Sigma^*)$ with $\Sigma^* = \Sigma(B^*, \Omega^*)$, where the function $\Sigma$ is given by
\begin{align}\label{eq:MCD}
    \Sigma(B, \Omega) = (I-B^\T)^{-1}\Omega(I-B)^{-1},
\end{align}
which is commonly referred to as the modified Cholesky decomposition. The decomposition is not unique, that is, there may exist multiple pairs \( (B', \Omega') \) such that \( \Sigma^* = \Sigma(B', \Omega') \). 
More precisely, for each ordering \( \sigma \in \mathbb{S}^p \), there exists a unique pair \( (B^*_\sigma, \Omega^*_\sigma) \) such that (i) \( B^*_\sigma \) is consistent with \( \sigma \), that is, \( |(B^*_\sigma)_{ij}| > 0 \) only if \( \sigma^{-1}(i) < \sigma^{-1}(j) \), and (ii) $\Sigma^* = \Sigma(B^*_\sigma, \Omega^*_\sigma)$~\citep[Lemma~C4]{zhou2021complexity}. Let 
$\mathcal{D}(\Sigma^*):=  \{(B^*_\sigma, \Omega^*_\sigma):  \Sigma^* = \Sigma (B^*_\sigma, \Omega^*_\sigma) \text{ for } \sigma \in \bbS^p\}$ be the collection of all such pairs. 
Let $G^*_\sigma$ denote the corresponding DAG of the weighted adjacency matrix $B^*_\sigma$, whose edge set is $\{(i,j): (B^*_\sigma)_{ij} \neq 0 \}$. We define $\Omega^*_\sigma = \diag(\omega_1^\sigma, \dots, \omega_p^\sigma$), and interpret $\omega_j^\sigma = \Var(X_j \mid X_{\Pred_{j}(\sigma)}) $ as the error variance of $\sX_j$ given that the underlying DAG is $G^*_\sigma$.


\subsection{The minimum-trace condition under weakly increasing error variance}

If we correctly identify the true pair $(B^*, \Omega^*)$ among the elements of $\mathcal{D}(\Sigma^*)$, we can reconstruct the true causal graph $G^*$ from $B^*$. One special case in which the true pair is identifiable occurs when the error variances are equal~\citep{peters2014identifiability}, that is, $\Omega^* = \omega^*I$ for some $\omega^* > 0$. One can readily see, by the inequality of arithmetic and geometric means, $\trace (\Omega^*_\sigma) \geq p \{\det(\Omega_\sigma^*)\}^{1/p} = p \omega^* = \trace(\Omega^*)$ for $\sigma \in \bbS^*$, and the equality only holds for $\sigma \in [\sigma^*]$. 
\cite{aragam2019globally} generalize the condition and 
propose computing a DAG 
that minimizes $\trace (\Omega_\sigma^*)$ over all $\Omega_\sigma^*$ such that $(B_\sigma^*, \Omega_\sigma^*) \in \mathcal{D}(\Sigma^*)$. We refer to such a DAG as the minimum-trace DAG.
\begin{definition}[Minimum-trace DAG]
    For a covariance matrix $\Sigma^*$, a minimum-trace 
    permutation is any permutation $\tau \in \arg\min_{\sigma \in \bbS^p} \trace (\Omega_\sigma^*)$, where $(B_\sigma^*, \Omega_\sigma^*) \in \mathcal{D}(\Sigma^*)$. The corresponding minimum-trace DAG is defined as $G^*_\tau$.
\end{definition}
One plausible justification for the minimum-trace DAG is that the model favors DAGs which minimize the total error variance $\sum_{j=1}^p \omega^\sigma_j$. This quantity corresponds to the total residual sum of squares, as estimated from the data. However, as multiple minimum-trace DAGs may exist in general, the issue of identifiability issue remains. As far as is known, the only explicit case that ensures identifiability is when the error variances are assumed to be equal. We extend identifiability under the minimum-trace objective to the case of weakly increasing variance cases. Specifically, the error variances are weakly increasing  with respect to $\sigma^* \in \bbS^p$, if $\omega^{\sigma^*}_{\sigma^*(1)}\leq \omega^{\sigma^*}_{\sigma^*(2)} \leq \dots \leq \omega^{\sigma^*}_{\sigma^*(p)}$.

\begin{theorem}\label{thm:increasing}
Consider the model~\eqref{eq:str_eq_individual} with the true causal ordering $\sigma^* \in \bbS^p$ and the true DAG $G^*$. Suppose that the error variances are weakly increasing with respect to $\sigma^*$. Then, for any $\sigma \in [\sigma^*]$, $\sigma$ is a minimum-trace permutation, and the corresponding DAG $G^*_{\sigma}$ is the unique minimum-trace DAG, with $G^*_{\sigma} = G^*$.
\end{theorem}
\begin{proof}
    See Section~\ref{subsec:thm1} in the Supplementary Material.
\end{proof}

The result substantially extends the range of settings in which the minimum-trace DAG attains identifiability.
To see this, we assume that the parameter pair $(B^*, \Omega^*)$ is drawn from a distribution that is absolutely continuous with respect to the Lebesgue measure, as is commonly assumed in the literature~\citep{spirtes2000causation}. The condition holds on a set of positive measure,  whereas the equal error variance assumption corresponds to a measure-zero subset of the parameter space.
Also, the condition may offer an appealing description of some systems where upstream
variables exhibit lower variability, while downstream variables tend to accumulate propagated uncertainty.

\begin{remark}[On existing identifiability results]
Identifiability under weakly increasing error variances has also been established through two approaches for recovering the true causal ordering~\citep{park2020identifiability}: \cite{chen2019causal} identify the first variable in the true ordering based on the fact that it always has the smallest marginal variance, then remove that node and recursively apply the procedure to the remaining subgraph. With a similar procedure, \cite{ghoshal2018learning} identify the last variable in the true ordering by using that it always has the largest conditional variance given remaining variables.
Theorem~\ref{thm:increasing} provides a complementary perspective that the ordering identified by these procedures corresponds to the minimum-trace permutation, that is, the permutation $\sigma$ that minimizes  the total residual sum of squares, whose population counterpart is the minimizer of $\trace (\Omega_\sigma^*)$ among all $(B_\sigma^*, \Omega_\sigma^*) \in \mathcal{D}(\Sigma^*)$.
\end{remark}

We introduce results at the sample level that follow from the identifiability. Let $X$ denote an $n \times p$ data matrix, each row of which is an independent copy of $\sX$. Let $\sigma^*$ be the true causal ordering, and suppose that the error variances are weakly increasing in $\sigma^*$. Then, we have $\min_{\sigma \notin [\sigma^*]}   \trace(\Omega_\sigma^*) > \trace(\Omega^*)$ by Theorem~\ref{thm:increasing}. We consider a slightly stronger condition on the gap between the two terms, which states that there exists a constant $\xi > 0$ such that
\begin{align}\label{eq:gap}
    \min_{\sigma \notin [\sigma^*]}   \trace(\Omega_\sigma^*) / \trace(\Omega^*) > 1 + \xi.
\end{align}
This ``gap'' condition has been employed in high-dimensional results~\citep{chang2024order, aragam2019globally}
, and is also called the ``omega-min'' condition in the equal variance setting~\citep{van2013ell}.

\begin{corollary}[Support recovery~\citep{aragam2019globally}]\label{cor:aragam}
    Consider the model described in Section~\ref{sec:aragam} of the Supplementary Material, and suppose that  assumptions (B1)-(B5) therein hold. Assume that  condition~\eqref{eq:gap} holds.
    Then,  $G^*$ can be identified with probability at least $1 - O(p^{-k})$, where $k = \max_{j \in [p]} |\mathrm{PA}_j(G^*)|$.
\end{corollary}

\begin{proof}
See Section~\ref{sec:aragam} in the Supplementary Material. 
\end{proof} 
\begin{corollary}[Strong model selection consistency~\citep{chang2024order}]\label{cor:woody}
    Let $\post$ denote the posterior distribution of the Bayesian order-based model described in Section~\ref{sec:woody} of the Supplementary Material, and suppose that assumptions (C1)-(C3) therein hold. Assume that condition~\eqref{eq:gap} holds.
    Then,     $\post(G^*)$ converges to 1.
\end{corollary}
\begin{proof}
See Section~\ref{sec:woody} in the Supplementary Material.  
\end{proof}

\section{Hill climbing algorithms on the permutation space}\label{sec:computation}

Given that identifiability is attained when the
error variances are weakly increasing in the true ordering, solving the following maximization problem
\begin{align}\label{eq:object}
    \hat{\sigma} = \arg \max_{\sigma \in \bbS^p} \{ - \trace (\Omega_\sigma)\}
\end{align}
is challenging due to the combinatorial nature of the permutation space $\bbS^p$. This is simply because enumerating all permutations requires $p!$ queries, which becomes computationally infeasible even for moderate values of $p$. 
Among various methods, hill climbing algorithms offer one of the simplest and most intuitive approaches to finding a solution. 
This class of algorithms proceeds by evaluating all states in a predefined local neighborhood $\cN(\sigma)$ of the current state $\sigma$, selecting the best one as the next state, and iterating this procedure until no further improvement is possible. 
We consider the class of hill climbing algorithms that use  $ - \trace(\Omega_\sigma)$ as the objective function. 
The choice of the local neighborhood $\cN$ significantly affects the efficiency of the algorithms. One of the most commonly used options is the adjacent transposition neighborhood~\citep{teyssier2005ordering,  agrawal2018minimal, friedman2003being}, defined as
\begin{align*}
    \cN_{\mathrm{ADJ}} = \{\sigma' \in \bbS^p \mid \sigma' = \mathrm{ADJ}(\sigma, i) \text{ for }i\in [p-1]\},
\end{align*}
where the adjacent transposition operator $\mathrm{ADJ}(\sigma, i)$ swaps the $i$-th and $(i+1)$-th elements of $\sigma$, 
\begin{align*}
    \mathrm{ADJ}(\sigma, i)\,:\, (\sigma(1), \dots, \sigma(i), \sigma(i+1),  \dots, \sigma(p))  \mapsto (\sigma(1), \dots, \sigma(i+1), \sigma(i),  \dots, \sigma(p)),
\end{align*}
for $i \in [p-1]$. This is because it offers
computational advantages: only $p-1$ evaluations are needed, and since each candidate differs only slightly, parts of the score calculation can be reused, making evaluation within each iteration efficient~\citep{chang2024order}. However, the search may easily get trapped in local optima due to its minimal search range. We introduce the random-to-random (R2R) operator
\begin{align*}
    \mathrm{R2R}(\sigma, i, j)\,:\, (\sigma(1), \dots, \sigma(i), \dots,  \sigma(j),  \dots, \sigma(p))  \mapsto (\sigma(1), \dots, \sigma(j), \sigma(i),  \dots, \sigma(p)),
\end{align*}
which outputs an ordering obtained by inserting the $j$-th element of $\sigma$ to the $i$-th position, defined for $i < j$. 
\begin{algorithm}[t!]
\caption{Hill climbing algorithm with R2R neighborhood} \label{alg:oracle}
\KwInput{An initial ordering $\sigma$, a covariance matrix $\Sigma$, and a R2R neighborhood $\mathcal{N}_{\mathrm{R2R}}$}
\While{TRUE}
{ 
$\tau \leftarrow \arg\max_{\sigma' \in \mathcal{N}_{\mathrm{R2R}} (\sigma)} \{-\trace(\Omega_{\sigma'}): (B_{\sigma'}, \Omega_{\sigma'}) \in \mathcal{D}(\Sigma)\}$ \\
\If{$\trace(\Omega_\tau) < \trace(\Omega_\sigma)$}{\vspace{1mm}
    $\sigma \leftarrow \tau$
}\Else{
    \textbf{Break}
}
\vspace{1mm}
}
\KwOutput{A DAG $G$ corresponding to the weighted adjacency matrix $B_\sigma$}
\end{algorithm}
Defining the R2R neighborhood as $\cN_{\mathrm{R2R}}(\sigma) = \{\sigma' \in \bbS^p \mid \sigma' = \mathrm{R2R}(\sigma, i, j) \text{ for }i < j, \,\, i, j \in [p]\}$, we prove that the steepest-ascent hill climbing algorithm with $\cN_{\mathrm{R2R}}$ (Algorithm~\ref{alg:oracle}) does not admit any strict local optima; that is, for any $\sigma \notin [\sigma^*]$,
there does not exist a case where $\min_{\tau \in \cN_{\mathrm{R2R}}(\sigma)} \trace(\Omega_\tau) > \trace(\Omega_\sigma)$.

\begin{theorem}\label{thm:computation}
Consider the model~\eqref{eq:str_eq_individual} with the true causal ordering $\sigma^* \in \bbS^p$ and the true DAG $G^*$. Suppose that the error variances are weakly increasing with respect to $\sigma^*$, and  
\begin{equation}\label{eq:assumption1}
    \Var(\sX_{\sigma^*(i)} \mid \sX_{\sigma^*(1)}, \cdots, \sX_{\sigma^*(i-1)}) \le \Var(\sX_{\sigma^*(j)} \mid \sX_{\sigma^*(1)}, \cdots, \sX_{\sigma^*(j-1)}, \sX_{\sigma^*(j+1)}, \cdots, \sX_{\sigma^*(p)})
\end{equation}
for all $i < j$. Then, Algorithm~1 does not admit any strict local optima.
\end{theorem}
\begin{proof}
    See Section~\ref{subsec:thm2} in the Supplementary Material. 
\end{proof}

As the weakly increasing error variance condition states that $\Var(\sX_{\sigma^*(i)} \mid \sX_{\sigma^*(1)}, \cdots, $ $ \sX_{\sigma^*(i-1)}) 
 \le \Var(\sX_{\sigma^*(j)} \mid \sX_{\sigma^*(1)}, \cdots, \sX_{\sigma^*(j-1)})$ for $i < j$, the assumption~\eqref{eq:assumption1} imposes a stronger condition. Assessing the assumption is challenging, as it requires analyzing the interplay between $B^*$ and $\Omega^*$ on a case-by-case basis. To evaluate how likely the result of Theorem~\ref{thm:computation} holds in practice,
 we conduct an extensive simulation study under standard settings. The algorithm admits no strict local optima, as the theorem states, and only a few weak local optima. See Section~\ref{subsec:validation} for further information. 

\begin{remark}[Comparison with other neighborhoods]
    The size of the R2R neighborhood is $p(p-1)/2$. We examine two other neighborhoods of equal size to see whether they also exhibit similar performance. The random transposition (RTS) operator is defined as
\begin{align*}
    \mathrm{RTS}(\sigma, i, j):\, (\sigma(1), \dots, \sigma(i), \dots,  \sigma(j),  \dots, \sigma(p))  \mapsto (\sigma(1), \dots, \sigma(j), \dots,  \sigma(i),  \dots, \sigma(p)),
\end{align*}
which corresponds to interchanging the $i$-th and the $j$-th elements of $\sigma$, while keeping the others unchanged. The reversed random-to-random (R2R-REV) operator is defined as
\begin{align*}
    \text{R2R-REV}(\sigma, i, j):\, (\sigma(1), \dots, \sigma(i), \dots,  \sigma(j),  \dots, \sigma(p))  \mapsto (\sigma(1), \dots, \sigma(j), \sigma(i),  \dots, \sigma(p)),
\end{align*}
which outputs the ordering obtained by inserting the $i$-th element of $\sigma$ to the $j$-th position with $i < j$. We defer the formal definitions of the ADJ, RTS, R2R, and R2R-REV operators, along with their examples, to Section~\ref{subsec:operators} in the Supplementary Material. Similar to ADJ and R2R neighborhoods, we define $\cN_{\mathrm{op}}(\sigma) = \{\sigma' \in \bbS^p \mid \sigma' = \mathrm{op}(\sigma, i, j) \text{ for }i < j, \,\, i, j \in [p]\}$, for $\mathrm{op} = \mathrm{RTS}$  and $ \text{R2R-REV}$.
Interestingly, the algorithms with both neighborhoods
admit several strict local optima and a significant number of weak local optima, indicating that they are suboptimal under the standard simulation setting. This suggests that neighborhoods of the same size can differ in performance, and careful analysis is required to design an algorithm effectively.
\end{remark}

    

\section{Simulation studies}

We use the following procedure to generate the true DAG \( G^* \) throughout the section. We fix the true ordering to be \(\sigma^* = (1, \ldots, p)\), and for each pair \((i, j)\) such that \(i < j\), we add edge \(i \rightarrow j\) to \(G^*\) with probability \(p_{\text{edge}} = 3 / (2p - 2)\). Hence, the expected number of edges of \(G^*\) is \(3p / 4\). We generate $\Sigma^* = \Sigma(B^*, \Omega^*)$ by generating edge weights \( B_{ij}^* \) for each edge \( i \rightarrow j \in G^* \) from the uniform distribution on \( [-1, -0.3] \cup [0.3, 1] \). The error variances $\{\Omega^*_{jj}\}_{j=1}^p = \{\omega^*_j\}_{j=1}^p$ are drawn from the uniform distribution on \( [1 - a, 1 + a] \), where \( a \sim \mathrm{Unif}[0, 1] \) and are then sorted in increasing order. We resample $(G^*, B^*, \Omega^*)$  for each replication in each simulation setting.

\subsection{Comparison of local neighborhoods}\label{subsec:validation}

In this simulation, we fix \( p = 8 \) so that we can search over all \( 8! = 40{,}320 \) possible orderings to count the number of strict and weak local optima. 
Given a search neighborhood \(\mathcal{N}\), we say that $\sigma \in \bbS^p$ is a strict local optimum if $\operatorname{tr}(\Omega_\sigma) < \operatorname{tr}(\Omega_{\sigma'})$ for all $\sigma' \in \mathcal{N}(\sigma)$. Similarly, $\sigma \in \bbS^p$ is a weak local optimum if $\operatorname{tr}(\Omega_\sigma) \leq \operatorname{tr}(\Omega_{\sigma'})$ for all $\sigma' \in \mathcal{N}(\sigma)$. 
To compare the four neighborhoods defined in Section~\ref{sec:computation}, we compute the number of strict and weak local optima, and average the results over 10,000 random seeds. The results are summarized in Table~\ref{tab:local_modes}.
We highlight that the R2R neighborhood produces no strict local optima, while the RTS and R2R-REV neighborhoods admit at least one strict local optimum in 34 and 61 out of 10{,}000 simulations, respectively. 
The RTS and R2R-REV neighborhoods produce weak local optima in 64\% and 53\% of the 10{,}000 simulations, respectively, whereas the R2R neighborhood yields only 4 such cases. The ADJ neighborhood, a standard choice for order-based methods, yields weak local optima in 99.9\% of the total simulations. 
The result supports the conclusion of Theorem~2 and guides the choice of neighborhood in local search algorithms, suggesting that any selected neighborhood should include the R2R neighborhood.

\begin{table}[t]
\centering
\begin{tabular}{lcccc}

 & ADJ & RTS & R2R-REV  & R2R \\

Strict local optima       & 0.15 $\pm$ 0.00 & 0.00 $\pm$ 0.02 & 0.01 $\pm$ 0.00 & 0 $\pm$ 0  \\
Weak local optima  & 9472.36 $\pm$ 56.50 & 190.43 $\pm$  3.33  & 503.62 $\pm$ 8.07 &  0.02 $\pm$ 0.01 \\

\end{tabular}
\caption{The number of strict and weak local optima across four neighborhoods over 10,000 repetitions. The maximum possible value of an entry is $8! = 40,320$. Each value represents mean $\pm$ one standard error.}
\label{tab:local_modes}
\end{table}

\subsection{Algorithm complexity}\label{subsec:complexity}

Empirical results in Section~\ref{subsec:validation} indicate that  the algorithm is consistent in most cases. We further investigate how the number of iterations required for convergence  scales with the number of variables \( p = 5, 10, 20, 50, 100 \), and we evaluates estimation accuracy using the edge difference between the estimated and true graphs. The data $X$ are sampled from $\mathrm{MVN}_p(0, \Sigma^*)$, with $n = 1,000$ samples. We run the finite-sample algorithm outlined in the Supplementary Material (Algorithm~\ref{alg:full}). The result over 50 repetitions with random initial ordering is presented in Table~\ref{tab:edge_steps}. Interestingly, the number of iterations required for convergence never exceeds $p-1$, which suggests that the algorithm is highly efficient in practice. Under some settings and assumptions, it may be possible to prove polynomial-time convergence of the algorithm. However, such a result would require a case-by-case analysis.

\begin{table}[h]
\centering
\begin{tabular}{cccccc}

$p$ & \(  5 \) & \(  10 \) & \(  20 \) & \(  50 \) & \( 100 \) \\

Edge difference     &       0 $\pm$ 0   &     0 $\pm$ 0    &     0 $\pm$ 0     &     0 $\pm$ 0    &    0 $\pm$ 0      \\
Mean &  1.68 $\pm$  0.11   &  3.08 $\pm$ 0.20     &  5.16   $\pm$   0.33  &   11.68 $\pm$ 0.53   &   21.18 $\pm$  0.90   \\
Max  &  4  &  6  &  10 &   22 &  44 \\

\end{tabular}
\caption{Edge difference and the mean (second row) and maximum (third row) number of iterations to termination, over 50 repetitions with random initialization across varying values of \( p \). Entries in the first two rows report the mean \( \pm \) one standard error.}
\label{tab:edge_steps}
\end{table}







\bibliographystyle{plainnat}
\bibliography{reference_arxiv}


\newpage
\appendix

\section*{Appendix}

\section{Proofs}
In this section, we provide the proof of Theorem \ref{thm:increasing} and Theorem \ref{thm:computation}. 
We first provide auxiliary results that are crucial for the proof of Theorem \ref{thm:increasing} and Theorem \ref{thm:computation}. 

\subsection{Auxiliary results}

We start with the majorization relationship between two sorted vectors. 
\begin{definition}[Majorization]\label{def: majorization}
    For two vectors $a = (a_1, a_2, \ldots, a_p) \in \bbR^p$ and $b = (b_1, b_2, \allowbreak \ldots, b_p) \in \bbR^p$, we write $a \succeq b$, or say that $a$ \textit{majorizes} $b$, if and only if the following three conditions are satisfied.
\begin{enumerate}
    \item [(a)] \textit{Sorted vectors:} $a_1 \leq \cdots \leq a_p$ and $b_1 \leq \cdots \leq b_p$.
    \item [(b)] \textit{Dominating  partial sums of the $k$ largest elements:} For all $k = 1, \dots, p-1$, 
    \begin{align*}
        \sum_{j = p - k + 1}^p a_j \ge \sum_{j = p - k + 1}^p b_j.
    \end{align*}
    \item [(c)] \textit{Identical total sums:} $\sum_{j =1}^p a_j = \sum_{j =1}^p b_j $.
\end{enumerate}
\end{definition}

Next, we introduce Karamata's inequality. 
\begin{lemma}[Karamata's inequality]\label{lem: Karamata's inequality}
    For $a, b \in \bbR^p$ satisfying $a \succeq b$, 
    \begin{equation*}
        \sum_{i=1}^p f(a_i) \ge \sum_{i=1}^p f(b_i)
    \end{equation*}
    holds for any convex function $f$. For a strictly convex function $f$, equality holds if and only if $a_i = b_i$ holds for all $1 \le i \le p$.
\end{lemma}
As a direct application of Karamata's inequality, we introduce the following lemma.

\begin{lemma}\label{lem: not minimum lemma}
    Let $\{a_i\}_{i=1}^p$ and $\{b_i\}_{i=1}^p$ be positive sequences satisfying 
    $\sum_{i=1}^p \log a_i = \sum_{i=1}^p \log b_i$,
    with $a_i \geq b_i$ for all $2 \leq i \leq p$. If $b_1$ is the smallest element in the sequence $\{ b_i\}_{i=1}^p$, $\sum_{i=1}^p a_i \geq \sum_{i=1}^p b_i$ holds.
    Here, the equality holds if and only if $a_i = b_i$ holds for all $ 1\le i \le p$.
\end{lemma} 
\begin{proof}
    Let $(a_1', \ldots, a_p')$ and $(b_1', \ldots, b_p')$ be the sorted vectors of $(a_1, \ldots, a_p)$ and $(b_1, \ldots, b_p)$ in a non-decreasing order, respectively. We shall show that
    \begin{equation*}
        (\log a_1', \ldots, \log a_p') \succeq (\log b_1', \ldots, \log b_p').
    \end{equation*}
    First, condition (a) of Definition \ref{def: majorization} directly holds. For $1 \le k \le p-1$, we have
    \begin{align*}
        \sum_{j=p-k+1}^p \log b_j' &= \max_{S \subset [p]: |S| = k} \sum_{s \in S} \log b_s = \max_{S \subset \{2, \ldots, p\}: |S| = k} \sum_{s \in S} \log b_s
        \\&\le \max_{S \subset \{2, \ldots, p\}: |S| = k} \sum_{s \in S} \log a_s \le \max_{S \subset [p]: |S| = k} \sum_{s \in S} \log a_s = \sum_{j=p-k+1}^p \log a_j'.
    \end{align*}
    The second equality holds because $b_1$ is the smallest element in $\{b_1, \ldots, b_p\}$.
    Lastly, $\sum_{i=1}^p \log b_i' = \sum_{i=1}^p \log a_i'$ holds from the assumption. Therefore, applying Lemma \ref{lem: Karamata's inequality} for $f(x) = \exp(x)$ on these two vectors, we obtain
    \begin{equation*}
        \sum_{j=1}^p a_j = \sum_{j=1}^p a_j' = \sum_{j=1}^p \exp(\log  a_j') \ge \sum_{j=1}^p \exp(\log b_j') = \sum_{j=1}^p b_j, 
    \end{equation*}
    concluding the proof. The equality condition is obtained from Lemma \ref{lem: Karamata's inequality}
\end{proof}

Next, we introduce differential entropy of a continuous random variable $\sX$. Combined with \eqref{eq: diff entropy of Gaussian}, differential entropy is helpful in simplifying the argument regarding the conditional variance of jointly Gaussian random variables.

\begin{definition}[Differential entropy]
    For a continuous random variable $\sX$ with probability density function $f_{\sX}$, differential entropy $h(\sX)$ is defined as follows.
    \begin{equation*}
        h(\sX) = - \int f_{\sX}(x) \log f_{\sX}(x) \, dx
    \end{equation*}
\end{definition}
For a Gaussian random vector $\sX \sim \mathcal{N}(\mu, \Sigma)$ where $\mu \in \bbR^p$ and $\Sigma \in \bbR^{p \times p}$,
\begin{equation}\label{eq: diff entropy of Gaussian}
    h(\sX) = 2^{-1} \log \big( (2\pi e)^p \det(\Sigma) \big)
\end{equation}
holds. For a Gaussian random vector, its differential entropy only depends on the variance. Similarly, the conditional differential entropy $h(\sX_1 \mid \sX_2)$ is defined as 
\begin{equation*}
    h(\sX_1 \mid \sX_2) = - \int f_{\sX_1, \sX_2} (x_1, x_2) \log f_{\sX_1 \mid \sX_2}(x_1 \mid x_2) \, dx_1 dx_2.
\end{equation*}
We introduce the following properties of differential entropy for continuous random variables $\sX_1, \ldots, \sX_p$.
\begin{lemma}[Properties of differential entropy]\label{lem: differential entropy properties}
    For continuous random variables $\sX_1, \ldots, \sX_p$, the following holds.
    \begin{enumerate}
    \item [(a)] \textit{Chain rule:} 
    \begin{equation*}
        h(\sX_1, \sX_2, \ldots, \sX_p)= \sum_{i=1}^p h(\sX_i \mid \sX_{i-1}, \ldots, \sX_1)
    \end{equation*}
    \item [(b)]
    \begin{equation*}
        h(\sX_1) \ge h(\sX_1 \mid \sX_2).
    \end{equation*}
\end{enumerate}
\end{lemma}
\begin{proof}
    We refer to Theorem 2.5.1 of \cite{cover1999elements}.
\end{proof}

From now on, we assume, without loss of generality, that $\sigma^* = \iota = (1, 2, \ldots, p)$ is the true ordering of the true causal DAG $G^*$ defined in \eqref{eq:str_eq_individual}. 
Recall that we define $\Pred_j(\sigma) = \{\sigma(k): k < \sigma^{-1}(j) \}$ as the set of predecessors of node $j$ under the ordering $\sigma$ and  $\omega^{\sigma}_j = \Var(\sX_{j} \mid \sX_{\Pred_j(\sigma)})$ for  each permutation $\sigma$.
Let $\omega_i = \omega_{i}^{\sigma^*}$ for notational convenience.
We introduce the following lemma, which plays a key role in the proof of Theorem~\ref{thm:increasing}.

\begin{lemma}\label{lem: log variance majorization}
    Assume that the error variances are weakly increasing in the true ordering, that is, $\omega_1 \le \omega_2 \le \cdots \le \omega_p$.
    For any given permutation $\sigma \in \bbS^p$, let $(v_1,\ldots, v_p)$ denote the vector obtained by sorting the set $\{\omega_1^{\sigma}, \omega_2^{\sigma}, \dots, \omega_p^{\sigma}\}$ in a non-decreasing order. Then, we have
    \begin{equation*}
        (\log v_1, \dots, \log v_p) \succeq (\log \omega_1, \dots, \log \omega_p).
    \end{equation*}
\end{lemma}
\begin{proof}
Our objective is to verify that all the conditions in Definition~\ref{def: majorization} are satisfied. By the definition of  $(v_1,\ldots, v_p)$ and the assumption on $(\omega_1,\ldots, \omega_p)$, condition (a) holds. Condition (c) follows from the observation that
\begin{align*}
    \sum_{j=1}^p \log v_j &= \sum_{j=1}^p 2h(\sX_j \mid \sX_{\Pred_j(\sigma)}) - \frac{p}{2}\log(2\pi e)
    \\& =2h(\sX_1, \ldots, \sX_p) - \frac{p}{2} \log (2\pi e)
    \\& = \sum_{j=1}^p 2h(\sX_{j} \mid \sX_{1}, \ldots, \sX_{j-1}) - \frac{p}{2}\log(2\pi e)
    \\& = \sum_{j=1}^p \log \omega_j.
\end{align*}
Here, we use~\eqref{eq: diff entropy of Gaussian}, and the second and the third equalities are obtained from Lemma \ref{lem: differential entropy properties} (a). Since $(\log v_1, \ldots, \log v_p)$ is the non-decreasing rearrangement of the set $\{\omega_1^{\sigma}, \omega_2^{\sigma}, \dots, \omega_p^{\sigma}\}$, we have
\begin{equation*}
    \sum_{j=p-k+1}^p \log v_j \ge \sum_{j=p-k+1}^p \log \omega^\sigma_j
\end{equation*}
for all $1 \le k \le p$. 
To show that condition (b) is satisfied, it suffices to verify
\begin{equation}\label{eq: log u log omega comparison}
    \sum_{j=p-k+1}^p \log \omega^\sigma_j \ge \sum_{p-k+1}^p \log \omega_j,
\end{equation}
for all $1 \le k \le p$. For a fixed $k \in [p]$,  let $\llow$ and $\lupp$  denote the smallest and the largest indices, respectively, at which the nodes in the set $R = \{p-k+1, \dots, p \}$ appear in a given ordering $\sigma$. For example, let $p = 6$, $k = 2$ and $\sigma = (3,6,4,1,5,2)$. Then, the node set $R$ is $\{5, 6\}$, and we have 
\begin{align*}
    \llow &= \min\{j : \sigma(j) = i, \text{ for } i \in R \} = 2, \\
    \lupp &= \max\{j : \sigma(j) = i, \text{ for } i \in R \} = 5.
\end{align*} 
We now define a permutation $\tau$ induced from $\sigma$ as follows. For all $i < \llow$ and $i > \lupp$, We set $\tau(i) = \sigma(i)$. Let $S = \{\sigma(i) :\llow  \le i \le \lupp\}$ denote the set of nodes appearing between positions $\llow$ and $\lupp$ in the ordering $\sigma$. Then, by definition, $R \subseteq S$. Let $L = S \setminus R$; then $L \subseteq \{1, \dots, p-k \}$. We define $\tau$ by rearranging the elements of $S$ so that all elements of $L$ appear before those of $R$, while preserving the relative order within each set. In the previous example, we have $S = \{1,4, 5, 6\}$ with $L = \{1, 4\}$, and $\tau = (3,4,1,6, 5, 2)$. See the illustration in Fig.~\ref{fig:tau}. 
\begin{figure}[h]
    \centering
    \includegraphics[width=0.98\linewidth]{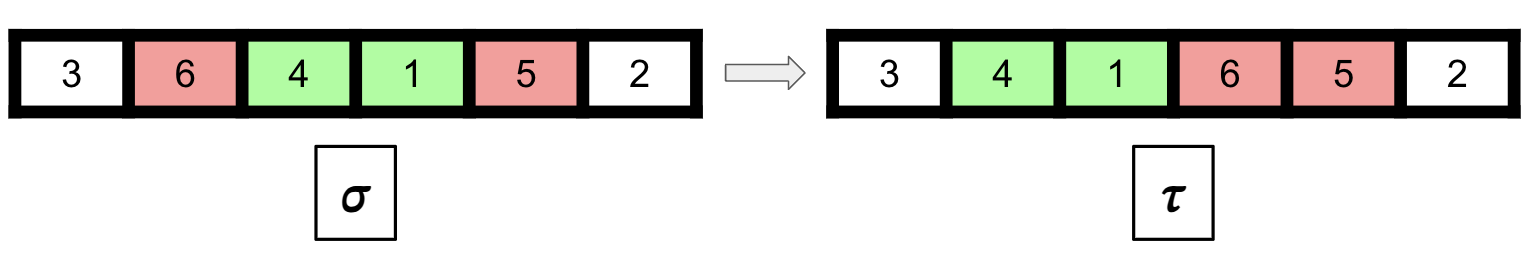}
    \caption{An illustration of the example with $p = 6$, $k = 2$ and $\sigma = (3,6,4,1,5,2)$. The set $L$ is colored green, and $R$ is colored red. }
    \label{fig:tau}
\end{figure}

For mathematical rigor, we formally define $\tau$ by these conditions
\begin{itemize}
    \item [1)] for $\llow \le i_1, i_2 \le\lupp$, $\sigma(i_1) < \sigma(i_2) \le p-k$, $\tau(i_1) < \tau(i_2)$, 
    \item [2)]  for $\llow \le i_1, i_2 \le\lupp$, $p-k+1 \leq \sigma(i_1) < \sigma(i_2)$, $\tau(i_1) < \tau(i_2)$, 
    \item [3)] for $\llow \le i_1, i_2 \le\lupp$, $\sigma(i_1) \le p-k < \sigma(i_2), \tau(i_1) < \tau(i_2)$.
\end{itemize}
Therefore, we can write $\tau(i)$ for $\llow \le i \le \lupp$ as follows.
\begin{align*}
    \tau \big (i - \sum_{p-k+1 \le a \le p} \ind_{\{\sigma^{-1}(a) \le i\}} \big) = \sigma(i), & \;\sigma(i) \le p-k,\\
    \tau \big (\lupp - k + \sum_{p-k+1 \le a \le p} \ind_{\{\sigma^{-1}(a) \le i\}} \big) = \sigma(i), & \;\sigma(i) \ge p-k+1,
\end{align*}
where $\ind$ denotes the indicator function.
By construction, $\Pred_{j}(\sigma) \subset \Pred_{j}(\tau)$ for $p-k+1 \le j \le p$. We show that equation~\eqref{eq: log u log omega comparison} holds by the following derivation.
\begin{align*}
    \sum_{j = p - k + 1}^p \log \omega^\sigma_j & = \sum_{j = p - k + 1}^p \log \Var(\sX_j \mid \sX_{\Pred_{j}(\sigma)}) \\
    & = \sum_{j = p - k + 1}^p 2h(\sX_j \mid \sX_{\Pred_{j}(\sigma)}) - k\log(2\pi e) \quad (\because \eqref{eq: diff entropy of Gaussian})\\
    & \ge \sum_{j = p - k + 1}^p 2h(\sX_j \mid \sX_{\Pred_{j}(\tau)}) - k\log(2\pi e) \quad (\because \text{Lemma \ref{lem: differential entropy properties} (b)})\\
    & = 2h \big(\sX_{p-k+1}, \sX_{p-k+2}, \ldots, \sX_p \mid \sX_1, \sX_2, \ldots, \sX_{p-k} \big) - k\log(2\pi e) \quad (\because \text{Lemma \ref{lem: differential entropy properties} (b)} )\\
    & = \sum_{j = p - k + 1}^p 2h(\sX_j \mid \sX_1, \sX_2, \ldots, \sX_{j-1}) - k\log(2\pi e) \quad (\because \text{Lemma \ref{lem: differential entropy properties} (a)})
    \\& = \sum_{j = p - k + 1}^p \log \mathrm{Var}(\sX_j | \sX_1, \dots, \sX_{j-1}) \\
    & = \sum_{j = p - k + 1}^p \log \omega_j,
\end{align*}
which concludes the proof.
\end{proof}

\subsection{Proof of Theorem~\ref{thm:increasing}}\label{subsec:thm1}

We introduce some notation. We say that $k_1 \rightarrow k_2 \rightarrow \cdots \rightarrow k_m $ is a direct path in $G^*$ if $k_{i+1} \in \Pa_{k_i}(G^*)$ for all $i \in [m-1]$. We define the descendent set of node $j$ in $G^*$ as
\begin{align*}
    \Des_j(G^*) = \{\ell: \text{ there exists a direct path from } j \text{ to } \ell \text{ in } G^* \}.
\end{align*}
Without loss of generality, we assume the true ordering $\sigma^* = \iota = (1, \dots, p)$ with $\omega_1 \le \dots \le \omega_p$. By combining the result of Lemma \ref{lem: log variance majorization} and Lemma \ref{lem: Karamata's inequality} with convex function $f(x) = \exp(x)$, we obtain
\begin{equation*}
    \sum_{j = 1}^{p} \omega_j = \sum_{j = 1}^{p} \exp (\log \omega_j) \leq \sum_{j = 1}^{p} \exp (\log v_j) = \sum_{j = 1}^{p} v_j = \sum_{j = 1}^{p} \omega_j^\sigma,
\end{equation*}
for any permutation $\sigma$. Therefore, $\sigma^*$ minimizes the trace. Furthermore, we can prove that $\sum_{j=1}^p \omega_j^{\sigma}$ is minimized if and only if $\sigma \in [\sigma^*]$, by using the equality condition in Lemma~\ref{lem: Karamata's inequality}, which implies that equality holds if and only if $\omega_j = v_j$ holds for all $j \in [p]$.
For any $\sigma \in [\sigma^*]$, we have $\sum_{j = 1}^{p} \omega_j  = \sum_{j = 1}^{p} \omega_j^\sigma $, since $\omega_j^{\sigma} = \omega_j$ for all $j \in [p]$. To see this, we have
\begin{align*}
    \omega_j^\sigma = \Var(\sX_j \mid \sX_{\Pred_j(\sigma)}) = \Var(\sX_j \mid \sX_{\Pa_j(G^*)})
\end{align*}
by the definition of a consistent ordering. We refer to any $\sigma \in [\sigma^*]$ as a minimum-trace permutation. 

On the other hand, for any $\sigma \notin [\sigma^*]$, there exists  an edge $j \rightarrow k \in G^*$ with $\sigma^{-1}(k) < \sigma^{-1}(j)$. We define 
\begin{align*}
    k^* = \max \{k \in [p]: \text{ there exists } j \text{ such that } j\to k \in G^*, \sigma^{-1}(k) < \sigma^{-1}(j)\}.
\end{align*}
Then, for $\ell > k^*$, we have $\omega_\ell = \omega_\ell^\sigma = v_\ell$. 
Now, we claim that $\omega_{k^*}  < \omega_{k^*}^\sigma$. 
We define another ordering $\sigma'$ such that $(\sigma')^{-1}(k^*) < (\sigma')^{-1} (\ell)$ only if $\ell \in \Des_{k^*}(G^*)$, that is, only descendant nodes of $k^*$ will appear after $k^*$ in $\sigma'$. This also means that the predecessors of node $k^*$ under the ordering $\sigma'$ are non-descendant nodes, that is,  
\begin{align*}
    \Pred_{k^*}(\sigma') = \NonDes_{k^*}(G^*) := [p] \setminus (\Des_{k^*}(G^*)\cup \{ k^*\} ).
\end{align*}
By the local Markov property, which states that each node is conditionally independent of its non-descendants given its parents, 
\begin{align*}
    \omega_{k^*}^{\sigma'} = \Var(\sX_{k^*} \mid  \sX_{\Pred_{k^*}(\sigma')}) = \Var(\sX_{k^*} \mid  \sX_{\NonDes_{k^*}(G^*) }) =  \Var(\sX_{k^*} \mid \sX_{\Pa_{k^*}(G^*)}) = \omega_{k^*}.
\end{align*}

Next, we show that $\Pred_{k^*}(\sigma) \subsetneq \Pred_{k^*}(\sigma')$.  To see this, observe that $\Des_{k^*}(G^*) \cap \Pred_{k^*}(\sigma) = \emptyset$; otherwise there would exist a node $k > k^*$ and  an edge $j \rightarrow k \in G^*$ with $\sigma^{-1}(k) < \sigma^{-1}(j)$. This would contradict the definition of $k^*$ as the maximum element. To see the inclusion is strict, we use the definition of $k^*$: there exists a node $j^*$ such that $j^* \rightarrow k^*$ with $ \sigma^{-1}(k^*) < \sigma^{-1}(j^*)$. This implies that  
$\Pred_{k^*}(\sigma)$ is missing 
at least  one parent of $k^*$; specifically, $\Pred_{k^*}(\sigma') \setminus \{j^*\} \subseteq \Pred_{k^*}(\sigma)$. Therefore,
\begin{align*}
    \omega_{k^*} = \Var(\sX_{k^*} \mid  \sX_{\Pred_{k^*}(\sigma')}) < \Var(\sX_{k^*} \mid  \sX_{\NonDes_{k^*}(G^*) \setminus \{j^*\}}) \leq  \Var(\sX_{k^*} \mid \sX_{\Pred_{k^*}(\sigma)}) = \omega_{k^*}^\sigma,
\end{align*}
where the first inequality holds because $j^*$ is a parent of $k^*$.

Finally, if two vectors are equal when sorted in increasing order, their $m$-th smallest elements must be identical for all $m \in [p]$. Using the previous result, we have $\omega_k = \omega_k^\sigma$, for $k > k^*$. The $k^*$-th smallest element $\omega_{k^*}^\sigma$ for the ordering $\sigma$ is strictly greater than $\omega_{k^*}$, and since $\omega_{k} < \omega_{k^*} $ for $k < k^*$, $\omega_{k}$ cannot be equal to $\omega_{k^*}^\sigma$. This  implies that the two sorted vectors are not identical. Therefore, for any $\sigma \notin [\sigma^*]$, $\sum_{j = 1}^{p} \omega_j  < \sum_{j = 1}^{p} \omega_j^\sigma $ by using the equality condition in Lemma~\ref{lem: Karamata's inequality}. We can readily verify that the weighted adjacency matrix $B^*_\sigma$ encodes the same DAG, so we obtain the unique minimum-trace DAG, $G^*_\sigma = G^*_{\sigma^*}$. By definition, $G^*_{\sigma^*} = G^*$, the true DAG. which completes the proof.

\subsection{Proof of Theorem~\ref{thm:computation}}\label{subsec:thm2}

Similar to the proof of Theorem~\ref{thm:increasing}, we assume the true ordering $\sigma^* = \iota = (1, \dots, p)$ with $\omega_1 \le \dots \le \omega_p$.
For a permutation $\sigma$, we consider the permutation $\tau = \RtoR(\sigma, i, j)$.  Then we obtain
\begin{equation*}
    f(\sigma) - f(\tau) = \sum_{k=i}^j \omega_{\sigma(k)}^\sigma - \omega_{\tau(k)}^{\tau},
\end{equation*}
since $\omega_{\sigma(k)}^\sigma = \omega_{\tau(k)}^{\tau}$ holds for all $k < i$ and $k > j$. 
Next, we define $a_k = \omega^{\sigma}_{\sigma(k)}$ and $b_k = \omega_{\sigma(k)}^{\tau}$ for $i \le k \le j$.
From the definition of $\tau$, 
\begin{align*}
    b_k & = \Var(\sX_{\sigma(k)} \mid \sX_{\sigma(j)}, \sX_{\Pred_{\sigma(k)}(\sigma)}) \text{ for } i \le k \le j-1, \\
    b_j &= \Var(\sX_{\sigma(j)} \mid \sX_{\Pred_{\sigma(i)}(\sigma)}).
\end{align*}
From Lemma \ref{lem: differential entropy properties} (a) and \eqref{eq: diff entropy of Gaussian}, we obtain
\begin{align*}
    \sum_{k=i}^j \log a_k = \sum_{k=i}^j \log w^\sigma_{\sigma(k)} &= 
    2 \sum_{k=i}^j h(\sX_{\sigma(k)} \mid \sX_{\Pred_{\sigma(k)}(\sigma)}) - (j - i + 1)\log (2\pi e)
    \\&= 
    2h(\sX_{\sigma(i)}, \dots, \sX_{\sigma(j)} \mid 
    \sX_{\sigma(1)}, \cdots, \sX_{\sigma(i-1)} ) - (j - i + 1)\log (2\pi e)
    \\&= 2 \sum_{k=i}^{j-1} h(\sX_{\sigma(k)} \mid \sX_{\Pred_{\sigma(k)}(\tau)}) + 2h(\sX_{\sigma(j)} \mid  \sX_{\Pred_{\sigma(j)}(\tau)}) - (j - i + 1)\log (2\pi e)
    \\&= \sum_{k=i}^j \log w^\tau_{\sigma(k)} = \sum_{k=i}^j \log b_k.
\end{align*}
Additionally, from Lemma \ref{lem: differential entropy properties} (b), we obtain $a_k \ge b_k$ for all $i \le k \le j-1$. Now, consider $\sigma \notin [\sigma^*]$. In this case, we can find an index $i$ such that $\sigma(i) > i$. Let $i^*$ be the first such node. Let $j^* = \sigma^{-1}(i^*)$. We show that for $\tau = \RtoR(\sigma, i^*, j^*)$, we have $f(\sigma) - f(\tau) \geq 0$. To this end, we verify that the assumptions required to invoke  Lemma~\ref{lem: not minimum lemma}. We have shown that $\sum_{k=i^*}^{j^*}\log a_k  = \sum_{k=i^*}^{j^*} \log b_k $, and $a_k \geq b_k$ for all $i^* \leq k \leq j^*-1$. Now it is suffice to show that $b_{j^*}$ is the smallest element among $b_k$ for $i^* \leq k \leq j^*-1$. For $i^* \leq k \leq j^*-1$, we have
\begin{align*}
     b_{j^*} &= \Var(\sX_{i^*} \mid \sX_{\Pred_{j^*}(\sigma)}) \\
     &= \Var(\sX_{i^*} \mid \sX_{1}, \dots, \sX_{i^*-1}) 
     \\
     & \leq \Var(\sX_{\sigma(k)} \mid \sX_{\sigma(1)}, \dots, \sX_{\sigma(k-1)}, \sX_{\sigma(k+1)}, \dots, \sX_{\sigma(p)}) \\
     & \leq \Var(\sX_{\sigma(k)} \mid \sX_{i^*}, \sX_{\Pred_{\sigma(k)}(\sigma)} ) = b_k.
\end{align*}
Here, the second inequality follows from the assumption on $i^*$, and the first inequalty follows from the assumption~\eqref{eq:assumption1}. The last inequality is from  the fact that conditioning on additional variables cannot increase the conditional variance. Applying Lemma~\ref{lem: not minimum lemma}, we have  $\sum_{k=i^*}^{j^*} a_k  \geq \sum_{k=i^*}^{j^*} b_k $, which implies that  $f(\sigma) -f(\tau) \geq 0$. Since we have identified an element $\tau \in \cN_{\RtoR}(\sigma)$, so $\sigma$ cannot be a strict local optimum under the hill climbing algorithm with the R2R neighborhood.

\section{Proof of Corollary~\ref{cor:aragam}}\label{sec:aragam}  

\noindent
\textbf{Model specification.}
We consider a setting considered in~\cite{aragam2019globally}. Let $\mathbb{D}_p$ be the collection of weighted adjacency matrices of $p$-vertex DAGs. Suppose that a random vector $\sX$ satisfies the structural equation model $\sX = (B^*)^\T \sX + \varepsilon$, where   $\varepsilon \sim \mathrm{MVN}_p (0, \Omega^*)$, $B^* \in \bbD_p$, and $ \Omega^*$ is a positive diagonal matrix of size $p \times p$ satisfying weakly increasing in its causal ordering. It follows that $\sX \sim \mathrm{MVN}_p(0, \Sigma^*)$ where $\Sigma^* = \Sigma(B^*, \Omega^*)$ as defined in~\eqref{eq:MCD}.
Moreover, the true underlying DAG $G^*$ is the one encoded by $B^*$.
Let $X \in \bbR^{n \times p}$ denote the data, where each row is an independent copy of $\sX$. We consider the following optimization problem
\begin{equation*}
    \hB \in \argmin_{B \in \bbD_p} Q(B), \quad Q(B) = \frac{1}{2n}\|X - XB\|_F^2 + \rho_{\lambda}(B),
\end{equation*}
where $\| \cdot \|_F$ denotes the Frobenius norm, and $\rho_{\lambda}$ is the minimax concave penalty~\citep{zhang2010nearly}. As stated in the main text, there may exist multiple pairs $(B', \Omega')$ that satisfy $\Sigma^* = \Sigma (B', \Omega')$. We define $\mathcal{D}(\Sigma^*) = \{(B', \Omega'):  \Sigma^* = \Sigma (B', \Omega')\}$ as the collection of all such pairs. In practice, to reduce computational complexity, we often estimate an undirected skeleton, which is referred to as a superstructure, and restrict the DAG candidates to those whose undirected edges are subsets of the superstructure~\citep{perrier2008finding}. We assume the scenarios that a superstructure $\Gamma$ of $G^*$ is available. Thus, the objective becomes
\begin{align*}
    \hB \in \argmin_{B \in \bbD_{\Gamma}} Q(B),
\end{align*}
where $\bbD_{\Gamma} = \{B \in \bbD_p: \mathrm{skeleton}(G) \subseteq \Gamma \text{ for the DAG } G \text{ encoded by } B\}.$

\noindent
\textbf{Assumptions for high-dimensional analysis.}
Let $s = s(\Gamma)$ denote the maximum degree of the superstructure $\Gamma$, and define $\eta = \gamma_1 [1 + 6 \kappa(\Sigma; s) \gamma_2]$, where
\begin{align*}
\gamma_1  &= 4 \sqrt{\frac{ s \log (3ep/s) + \log p }{n}}, \\
\gamma_2  &= \left( 1 + 3 \sqrt{ \frac{ 2s \log (ep/s) }{n} } \right)^2, \\
\kappa(\Sigma; s) & = \frac{\sup_{S: |S| = 2s+2} \lmax(\Sigma^*_{SS})}{\inf_{S: |S| = s+1} \lmin(\Sigma^*_{SS})}, 
\end{align*}
where $\lmax(A)$ and $\lmin(A)$ denote the largest and smallest eigenvalues of the matrix $A$, respectively. Here is the list of assumptions. 
\vspace{1mm}

(B1) (Restricted eigenvalue condition) All eigenvalues of $\Sigma^*$ are bounded between constants $\underline{\nu}$ and $\overline{\nu}$,
\begin{align*}
    0< \underline{\nu} \leq \lmax(\Sigma^*) \leq \lmax(\Sigma^*)  \leq  \overline{\nu} < \infty.
\end{align*}

(B2) (Sample size) $s \log(p/s)+ \log p = o(n).$

(B3) (Regularization parameter)  $\lambda \gtrsim \sqrt{\log p/n}$ and  $\lambda > \eta$.

(B4) (Sparsity of the true graph) $|G^*| \lesssim p\sqrt{n/(s \log(p/s)+ \log p)}.$ 

(B5) (Beta-min condition) $\min\{|B^*_{ij}| : B^*_{ij} \neq 0\} \gtrsim \lambda $. 

\begin{proof}[Proof of Corollary~\ref{cor:aragam}]
Define the quantity
\[
\mathrm{gap}(\Sigma^*) := \inf \left\{ \operatorname{tr} \Omega' - \operatorname{tr} \Omega^* : \Omega' \neq \Omega^*,\ (B',  \Omega') \in \mathcal{D}(\Sigma^*) \right\}.
\]
By the gap condition in~\eqref{eq:gap}, $\mathrm{gap}(\Sigma^*) \ge \underline{\nu} \xi p$. With the assumption (B4), it satisfies Condition~3.1 of~\citet{aragam2019globally}. By Theorem~3.1 of~\citet{aragam2019globally}, we conclude the proof.    
\end{proof}

\section{Proof of Corollary~\ref{cor:woody}}\label{sec:woody}  

\noindent
\textbf{Model specification.} We first describe the notation and  setting for the Bayesian structure learning problem considered in~\cite{chang2024order}.  For each \(\sigma \in \mathbb{S}^p\) and \(G \in \mathcal{G}(\sigma)\), consider the structural equation model for the random vector \(\mathsf{X} = (\mathsf{X}_1, \ldots, \mathsf{X}_p)\)
\begin{align*}
    \mathsf{X}_j = \sum_{i \in \Pa_j(G)} B_{ij} \,\mathsf{X}_{i} + \varepsilon_j, \quad \varepsilon_j \mid \omega \overset{\text{i.i.d.}}{\sim} \mathcal{N}(0, \omega) \quad \text{for } j = 1, \ldots, p,
\end{align*}
where \(\Pa_j(G) \subseteq \Pred_j(\sigma)\) for each \(j\). Here, $\Pred_j(\sigma) = \{i \in [p]: \sigma^{-1}(i) < \sigma^{-1}(j) \}$ denotes the set of predecessors of node $j$ under the ordering $\sigma$. The following prior on the parameters \((\sigma, G, B, \omega)\) is used, where \(\pi_0\) denotes the prior density function
\begin{align*}
    B_{\Pa_j(G),j} \mid G, \omega & \overset{\text{ind}}{\sim} \mathcal{N}_{|\Pa_j(G)|} \left( 
\widehat{B}_{\Pa_j(G),j}, \frac{\omega}{\gamma} \left( X_{\Pa_j(G)}^\top X_{\Pa_j(G)} \right)^{-1}
\right), \quad \forall j \in [p], \\ 
\pi_0(\omega \mid \sigma) & \propto \omega^{-\frac{\kappa}{2} - 1}, \\
\pi_0(G, \sigma) & \propto \left(p^{c_0} \right)^{-|G|} \ind_{\{\widehat{G}_\sigma\}}(G),
\end{align*}
where \(\widehat{B}_{\Pa_j(G),j}\) is the least-squares estimator of \(B_{\Pa_j(G),j}\), and \(c_0, \gamma, \kappa\) are hyperparameters of the prior. The posterior distribution of $(G, \sigma)$ is given by
\begin{align*}
    \pi_n(G, \sigma) & \propto 
    \,  \pi_0(G, \sigma) \int \pi_0(B, \omega \mid G, \sigma) \, L(B, \omega)^\alpha \, d(B, \omega)\\
    & = e^{\phi(G)} \, \bm{1}_{\{\widehat{G}_\sigma\}}(G),
\end{align*}
where the $\alpha$-likelihood function $L(B, \omega)^\alpha$ with $\alpha \in (0,1)$ is used to offset the influence of the data under the empirical prior. We refer to $\phi(G)$ as the score of $G$, which is given by
\begin{align*}
    \phi(G) = -|G| (c_0 \log p 
+ 0.5\log[(1 + \alpha / \gamma)]) 
- \frac{\alpha p n + \kappa}{2} 
\log \left( \sum_{j=1}^{p} X_j^\top \Phi_{\mathrm{Pa}_j(G)}^\perp X_j
 \right),
\end{align*}
where $\Phi_S^\perp = I - X_S(X_S^\top X_S)^{-1} X_S$. Among the possible candidates for $\widehat{G}_\sigma$, the maximum a posteriori (MAP) estimator is selected as 
$\widehat{G}^{\mathrm{MAP}}_\sigma= \arg\max_{G \in \mathcal{G}_{d_{\text{in}}}(\sigma)} \phi(G)$, where $\mathcal{G}_{d_{\text{in}}}(\sigma) = \{G \in \mathcal{G}(\sigma): |\Pa_j(G)| \leq d_{\text{in}} \text{ for } j \in [p] \}$ denotes the set of DAGs consistent with ordering $\sigma$, subject to an in-degree constraint $d_{\text{in}}$. 

Let $G^*$ be the true underlying DAG and $B^*$ be a coefficient matrix such that each element $B_{ij}^*$ is nonzero if and only if there is an edge $i \rightarrow j$ in $G^*$. We say that $B^*$ is consistent with $G^*$. Let $[\sigma^*] \subseteq \bbS^p$ denote the set of orderings consistent with the true DAG $G^*$. Observe that $[\sigma^*]$ is nonempty due to acyclicity. Let $X \in \bbR^{n \times p}$ denote the data matrix, where each row is an independent copy of a random vector $\sX$. The distribution of $\sX$ is defined by the structural equation model  $\sX = (B^*)^\T \sX + \varepsilon$, where $\varepsilon \sim \mathrm{MVN}_p (0, \Omega^*)$, and $ \Omega^*$ is a positive diagonal matrix whose entries are weakly increasing with respect to the causal ordering. Let $\Sigma^* = \Sigma(B^*, \Omega^*)$ denote the covariance matrix of $\sX$.

\noindent
\textbf{Assumptions for high-dimensional analysis.}
We introduce
\[
d^* = \max_{(B', \Omega') \in \mathcal{D}(\Sigma^*)} \max_{j \in [p]} \{ |\Pa_j(G')|: B' \text{ is consistent with } G' \},
\]
which represents the maximum in-degree among all DAGs consistent with $\Sigma^*$. Here is the list of assumptions. \vspace{1mm}

(C1) (Restricted eigenvalue condition) There exist \(\underline{\nu}, \overline{\nu} > 0\) and a universal constant \(\delta > 0\) such that any eigenvalues of $\Sigma^*$ are bounded between $\underline{\nu}(1 - \delta)^{-2} $ and $ \overline{\nu}(1 + \delta)^{-2}$.

    (C2) (Hyperparameters) Assume \(\max \{ d^* , \max_j |\Pa_j(G^*)|\} \leq d_{\text{in}}\) where the sparsity parameter \(d_{\text{in}}\) satisfies \(d_{\text{in}} \log p = o(n)\), and prior parameters satisfy that \(\kappa \leq np\), \(0 \leq \alpha/\gamma \leq p^2 - 1\), \(c_0 > \rho(\alpha + 1) \max_{i \neq j}(\omega_j^*/\omega_i^*)\), and \(\rho > 4d_{\text{in}} + 6\).

    (C3) (Beta-min condition) $\min\{ |(B^*)_{ij}|^2 : (B^*)_{ij} \neq 0 \} \geq 16 c_0 \overline{\nu}^2 \log p / (\alpha \underline{\nu}^2 n).$

\begin{proof}[Proof of Corollary~\ref{cor:woody}]
We directly apply the result of Theorem~1 in \cite{chang2024order} by verifying Assumption~A and B therein. Assumption~A is satisfied by the gap condition in~\eqref{eq:gap}. Proposition~1 in \cite{chang2024order} ensures $\widehat{G}^{\mathrm{MAP}}_\sigma = G^*$ for $\sigma \in [\sigma^*]$ for sufficiently large $n$, with probability at least $1-4p^{-1}$ under conditions (C1)-(C3), therefore, Assumption~B is satisfied. Lastly, checking $d^* \leq \din$ and $\din \log p = o(n)$, as given by (C2), completes the proof.
\end{proof}

\newpage
\section{Algorithms}

\subsection{Local operators}\label{subsec:operators}

We formally define the types of local neighborhoods presented in the main text. 
 Let \((\cdot)_c\) denote an ordering in the cycle notation; for example, \(\mu = (a, b, c)_c\) is the ordering given by \(\mu(a) = b, \mu(b) = c, \mu(c) = a\) and \(\mu(k) = k\) for every \(k \notin \{a, b, c\}\). Let \(\circ\) denote the composition of two orderings; that is, \(\tau = \sigma \circ \mu\) is defined by \(\tau(i) = \sigma(\mu(i))\). 
Define the ADJ and RTS operators as
\begin{align*}
    \mathrm{ADJ}(\sigma, i) & = \sigma  \circ (i, i + 1)_c, \text{ for } i \in [p-1],\\
    \mathrm{RTS}(\sigma, i, j) & = \sigma  \circ (i, j)_c, \text{ for } i \neq j,  \,\, i, j \in [p], 
\end{align*}
respectively. See an example in Fig.~\ref{fig:adjrts}.
\begin{figure}[h!]
    \centering
    \centering
    \includegraphics[width=0.49\linewidth]{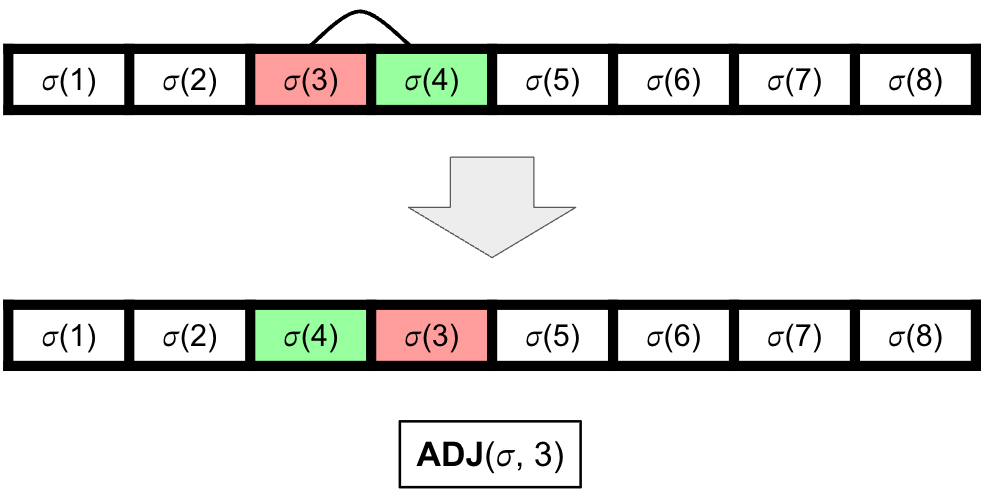}
    \includegraphics[width=0.49\linewidth]{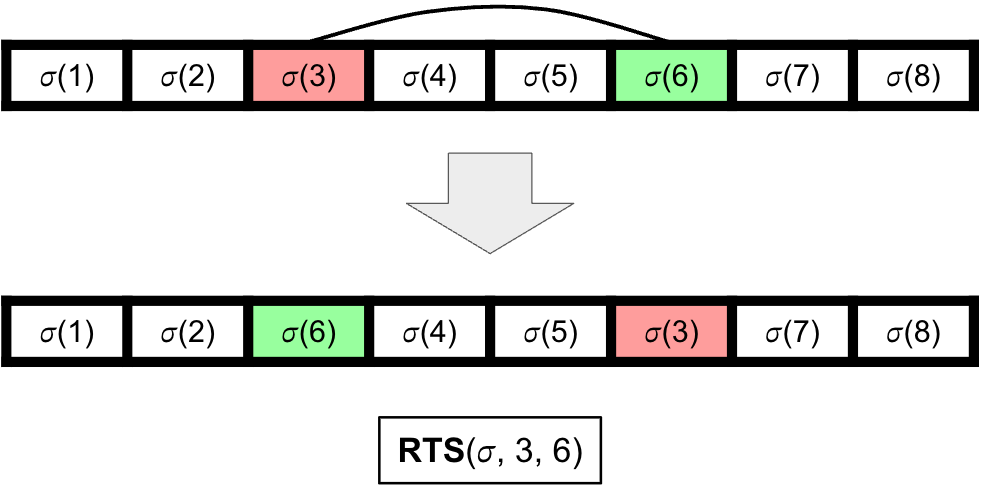}
    \caption{Example of ADJ($\sigma$, 3) and RTS($\sigma$, 3, 6): ADJ($\sigma$, 3) swaps the 3rd variable $\sigma(3)$ (in red) and the next (4th) variable $\sigma(4)$ (in green); RTS($\sigma$, 3, 6) interchanges the 3th variable $\sigma(3)$ (in green)  and the 6th variable $\sigma(6)$ (in green).}
    \label{fig:adjrts}
\end{figure}

For $i < j,  \,\, i, j \in [p]$, the R2R and R2R-REV operators are defined as
\begin{align*}
    \mathrm{R2R}(\sigma, i, j) & = \sigma  \circ  (i, i + 1, \ldots, j)_c,  \\
    \text{R2R-REV}(\sigma, i, j) & = \sigma  \circ (i, j, j-1, \ldots, i + 1)_c,
\end{align*}
respectively. See an example in Fig.~\ref{fig:r2r}. We note that the $insertion$ operator in~\cite{scanagatta2017improved} is the union of R2R and R2R-REV operator.
\begin{figure}[h!] \includegraphics[width=0.49\linewidth]{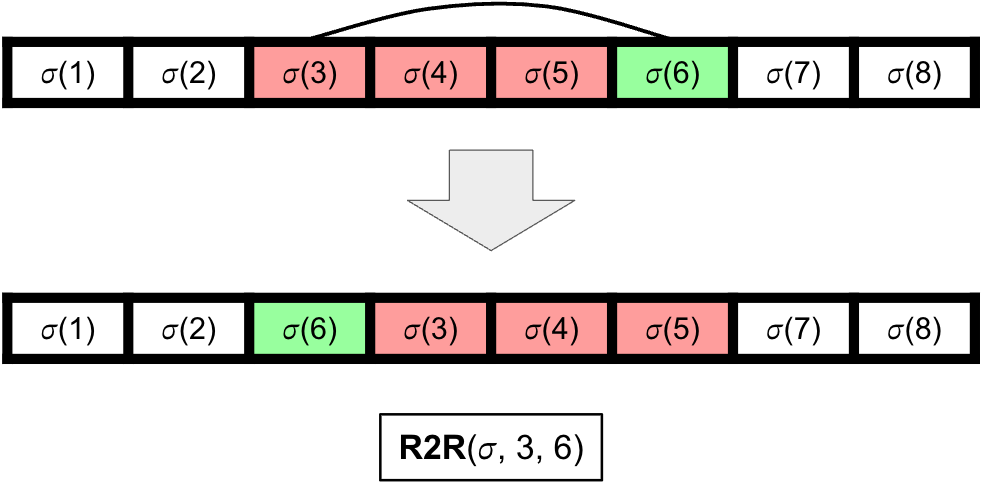}
    \includegraphics[width=0.49\linewidth]{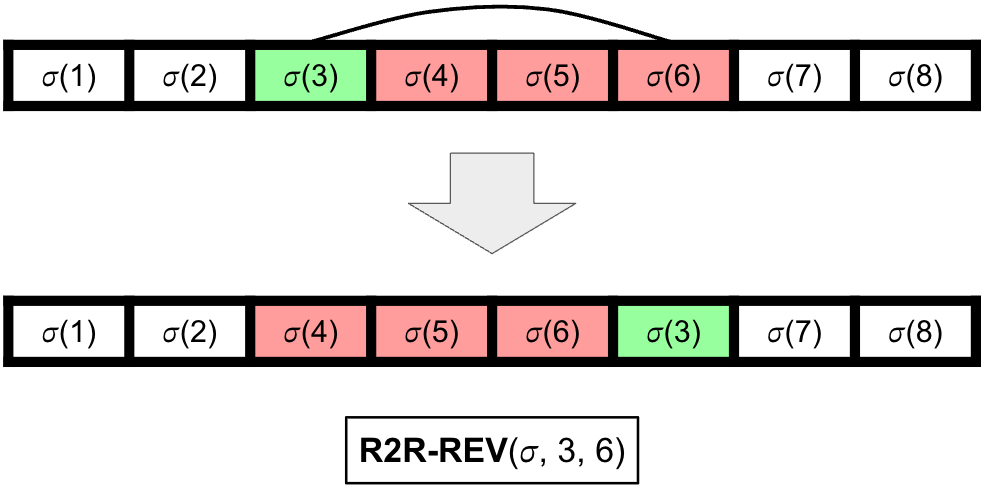}
    \caption{Example of R2R($\sigma$, 3, 6) and R2R-REV($\sigma$, 3, 6): R2R($\sigma$, 3, 6) inserts  the 6th variable $\sigma(6)$ (in green) into the 3rd position of the ordering $\sigma$; R2R-REV($\sigma$, 3, 6) inserts the 3th variable $\sigma(3)$ (in green) into the 6rd position of the ordering $\sigma$.}
    \label{fig:r2r}
\end{figure}

\subsection{Algorithm used in Section~\ref{subsec:complexity}}

For the finite sample algorithm, we define a score $\phi$ on $\sigma \in \bbS^p$ by using the model specification in Section~\ref{sec:woody}, which is given by 
\begin{align*}
    \phi(\sigma) = -|\widehat{G}^{\mathrm{MAP}}_\sigma| (c_0 \log p 
+ 0.5\log[(1 + \alpha / \gamma)]) 
- \frac{\alpha p n + \kappa}{2} 
\log \left( \sum_{j=1}^{p} X_j^\top \Phi_{\mathrm{Pa}_j(\widehat{G}^{\mathrm{MAP}}_\sigma)}^\perp X_j
 \right),
\end{align*}
where we use $c_0 = 3$, $\alpha = 0.99$, $\gamma = 0.01$, and $\kappa = 0$ for the hyperparameters. We outline the finite-sample hill climbing algorithm with R2R neighborhood in Algorithm~\ref{alg:full}.

\begin{algorithm}[h!]
\caption{Hill climbing algorithm with R2R neighborhood with data $X$} \label{alg:full}
\KwInput{An initial ordering $\sigma$,  data $X$, a R2R neighborhood $\mathcal{N}_{\mathrm{R2R}}$, and a score $\phi$}
\While{TRUE}
{ 
$\tau \leftarrow \arg\max_{\sigma' \in \mathcal{N}_{\mathrm{R2R}} (\sigma)} \phi(\sigma')$ \\
\If{$\phi(\tau) > \phi(\sigma)$}{\vspace{1mm}
    $\sigma \leftarrow \tau$
}\Else{
    \textbf{Break}
}
\vspace{1mm}
}
\KwOutput{An estimated DAG $\widehat{G}_\sigma$ given ordering $\sigma$.}
\end{algorithm}

\end{document}